\documentclass[letterpaper,10pt,conference]{ieeeconf}

\IEEEoverridecommandlockouts
\newtheorem{theorem}{Theorem}

\newtheorem{remark}{Remark}
\newtheorem{lemma}{Lemma}
\newtheorem{corollary}{Corollary}

\usepackage{amssymb}
\usepackage{graphicx}
\usepackage[cmex10]{amsmath}
\usepackage{array}
\usepackage{mdwtab}
\usepackage{graphics}
\usepackage{epsfig}
\usepackage{times}
\usepackage{amsmath}
\usepackage{amssymb}
\usepackage{url}
\usepackage{lscape}
\usepackage{color}
\usepackage{epsfig}
\usepackage{cite}
\usepackage{balance}
\usepackage{dblfloatfix}
\usepackage{hyperref}
\usepackage{mathtools}
\usepackage{mathrsfs}

\title{\LARGE \bf Physical Watermarking for Replay Attack Detection in Continuous-time Systems\thanks{This work is supported by the National Science Foundation, under grant number 1932530.}}

\author{Bahram Yaghooti, Raffaele Romagnoli, Bruno Sinopoli
\thanks{B. Yaghooti and B. Sinopoli are with the Department of Electrical and Systems Engineering, Washington University in St. Louis, St. Louis, MO, USA 63130 (Email: byaghooti@wustl.edu; bsinopoli@wustl.edu).
}
\thanks{R. Romagnoli is with the Department of Electrical and Computer Engineering, Carnegie Mellon University, Pittsburgh, PA, USA 15213 (Email: rromagno@andrew.cmu.edu).
}
}

\begin{document}

\maketitle
\thispagestyle{empty}
\pagestyle{empty}

\maketitle

\begin{abstract}
Physical watermarking is a well established technique for replay attack detection in cyber-physical systems (CPSs). Most of the watermarking methods proposed in the literature are designed for discrete-time systems. In general real physical system evolve in continuous time. In this paper, we analyze the effect of watermarking on sampled-data continuous-time systems controlled via a Zero-Order Hold. We investigate the effect of sampling on detection performance and we provide a procedure to find a suitable sampling period that ensures detectability and acceptable control performance. Simulations on a quadrotor system are used to illustrate the effectiveness of the theoretical results.
\end{abstract}

%\begin{IEEEkeywords}
%Fractional-Order Systems, Semilinear Systems, Pharmacokinetic Model, Constrained Control, Explicit Reference Governor.
%\end{IEEEkeywords}
%---------------------------------------------------------------------------------------------------------
%---------------------------------------------------------------------------------------------------------

\section{Introduction}\label{sec:introduction}
Cyber-Physical Systems (CPSs) integrate communication, computation, and control into physical world. CPSs play a crucial role in the design of efficient and sustainable services that are pillars of modern societies, such as energy delivery, transportation, health care, and water distribution \cite{lee2008cyber}. Their safety and security represent one of the main design challenges, as their heterogeneous and distributed nature makes CPSs vulnerable to a multitude of cyber-attacks \cite{markoff2010silent, sanger2012obama}.

The problem of cyber-attacks in networked control systems has been studied comprehensively in previous research \cite{khojasteh2020learning,kafash2018constraining, ghaderi2020blended,renganathan2020distributionally,khojasteh2019authentication}. In this paper our focus is on the analysis and detection of \textit{replay attacks} in control systems \cite{langner2011stuxnet,case2016analysis}. In such an attack, a cybercriminal eavesdrops on a network, and then fraudulently delays or resends old observations to misdirect the receiver into thinking that the system is behaving normally while carrying out their attack. One of the characteristics of replay attack which makes it simple to implement is that it can be used by a hacker without advanced knowledge of the system or skills to decrypt messages. This type of attack can be successful just by repeating a set of recorded data. 

% The principle behind the Replay Attack is simple. In such an attack, a cybercriminal eavesdrops on a network, and then fraudulently delays or resends a valid data to misdirect the receiver into doing what the attacker wants. One of the characteristics of replay attack which makes it to be used frequently is that it can be used by a hacker without advanced skills to decrypt a message. This type of attack can be successful just by repeating a set of recorded data.  

The first model of replay attacks on control systems together with a proposed countermeasure was introduced by Mo \textit{et al.} \cite{mo2009secure} and refined in subsequent papers \cite{mo2011cyber,mo2013detecting}. The basic idea revolves around the use of physical watermarking, a secret noisy control input added to an intended control input and aimed at authenticating the received observation. This framework was developed for discrete-time systems, and its performance has been studied extensively in the literature \cite{romagnoli2019model,hosseinzadeh2019feasibility,griffioen2019tutorial}.

%What is the effect of watermarking on the actual continuous-time system?  Considering the physical process as a stochastic continuous-time system, we show that the effectiveness of the discrete-time watermarking approach depends on the chosen sampling period.

%However, most of real physical systems are in continuous time. 
%The main objective of this research is to study the effectiveness of the watermarking method in continuous-time systems. In order to apply the developed framework to such systems, we need to discretize the main system and then design the controller and detector for the system. 
% In discrete-time, the effectiveness of the watermarking approach is based on the assumption that process and measurement noises are independent. Sampling a stochastic continuous-time system with independent process and measurement noises produces an equivalent sampled-data system with dependent previous noises.
% Even if we assume that the process noise and measurement noise are independent in continuous-time stochastic systems, they become dependent in the equivalent sampled-data system. This questions the credibility of the watermarking framework which is developed based on independent process and measurement noises. 

In general, many physical processes are continuous-time and controlled by sampling outputs and using a Zero-Order-Hold (Z.O.H) method for control. 
In this paper, we investigate the application of the watermarking framework to continuous-time systems and analyze the effect of sampling period on its performance. Specifically, while it is known that decreasing sampling period improves the performance of the controller, we will show that sampling also affects the performance of the detector. In this paper we explore this tradeoff to design an optimal sampling period. Finally, to illustrate the theoretical results, we apply the proposed methodology to a quadrotor hovering around an equilibrium point.

The rest of the paper is organized as follows: In section \ref{sec:problem_statement}, the discretization of a linear continuous-time stochastic systems is reviewed. Section \ref{sec:watermarking_disc_sys} provides basic concepts of the watermarking framework. Section \ref{sec:watermarking_cont_sys} investigates the effect of the sampling period on the controller and detector. It is shown that in the case of physical watermarking, the sampling period $T$ becomes a design parameter alongside the covariance of the watermarking signal. We therefore generalize the design of the watermarking signal in \cite{mo2013detecting} by defining a new optimization problem where we jointly design both the covariance of the watermarking and the sampling period, with appropriate constraints on the loss of performance and the maximum allowable sampling period as per \cite{franklin1998digital}.  To evaluate the theoretical results, the watermarking method is applied to a quadrotor, and simulation results are provided in section \ref{sec:simulation}. Finally, a brief conclusion is presented in the last section.
\section{System Description}\label{sec:problem_statement}
%In this paper, we study the feasibility of using watermarking framework for continuous-time systems.
Consider the following linear continuous-time stochastic system:
\begin{align}
    \dot{x}(t) & = Ax(t) + Bu(t) + w(t)\label{eq:linear_cont_sys} \\
    y(t) & = Cx(t) + v(t), \label{eq:output_cont_sys}
\end{align}
where $x\in\mathbb{R}^n$ is the system state vector;  $A\in\mathbb{R}^{n\times n}$, $B\in\mathbb{R}^{n\times p}$, and $C\in\mathbb{R}^{m\times n}$ are known and constant matrices; $u(t)\in\mathbb{R}^p$ and $y(t)\in\mathbb{R}^m$ are the input and output of the system, respectively; and $w(t)\in\mathbb{R}^n$ and $v(t)\in\mathbb{R}^m$ are the process and measurement noises, respectively. It is assumed that $w(t)$ and $v(t)$ are zero-mean Gaussian white noises,
\begin{align}
    \mathbb{E}[w(t)] & = 0 \\
    \mathbb{E}[v(t)] & = 0 \\
    \mathbb{E}[w(t)w^T(s)] & = Q\delta(t-s) \\
    \mathbb{E}[v(t)v^T(s)] & = R\delta(t-s), \label{eq:cont_measurement_cov}
\end{align} 
where $\delta(\cdot)$ is Dirac delta function, and $Q\in\mathbb{R}^{n\times n}$ and $R\in\mathbb{R}^{m\times m}$ are known matrices \cite{astrom1971introduction}.
% \textcolor{red}{put a reference to astrom or a footnote explaining the choices}

It is assumed that the process and measurement noises are independent of the previous and current state and independent of each other:
\begin{align}
    \mathbb{E}[w(t) x^T(s)] & = 0 \quad \text{for } t\geq s \\
    \mathbb{E}[v(t) x^T(s)] & = 0 \quad \text{for } t\geq s \\
    \mathbb{E}[w(t) v^T(s)] & = 0 \quad \text{for all } t,s \geq 0.
\end{align}
Physical watermarking has been used for replay attack detection in control systems. As this framework is developed for discrete-time systems, we first need to discretize system \eqref{eq:linear_cont_sys}. To this end, the control input will be applied using Zero-Order-Hold (ZOH) method with sampling period $T$, so the controller will produce a piece-wise constant command between sampling periods,
\begin{align}
    u(t) = u_k, \quad kT \leq t < (k+1)T, \quad  k=0,1,\cdots .
\end{align}
By integrating both sides of \eqref{eq:linear_cont_sys}, discrete-time system can be obtained as follows:
\begin{align}
    x((k+1)T) = & e^{AT}x(kT) \nonumber \\
    & + \Big(\int_{kT}^{(k+1)T}e^{A((k+1)T-\tau)}d\tau\Big)B u_k \nonumber \\
    & + \int_{kT}^{(k+1)T}e^{A((k+1)T-\tau)}w(\tau) d\tau .
\end{align}
Let us define $A_d$, $B_d$, and $w_k$ as 
\begin{align}
    A_d & = e^{AT} = \sum_{m=0}^{\infty}\frac{A^mT^m}{m!} \\
    B_d & = \Big(\int_{kT}^{(k+1)T}e^{A((k+1)T-\tau)}d\tau\Big)B = \sum_{m=0}^{\infty}\frac{A^m B T^{m+1}}{(m+1)!} \label{eq:B_discrete_2} \\
    w_k & = \int_{kT}^{(k+1)T}e^{A((k+1)T-\tau)}w(\tau) d\tau .
\end{align}
% By using Taylor expansion, we can rewrite matrix $A_d$ and $B_d$ in the following form:
% \begin{align}
%     A_d & = I + AT + \frac{1}{2!}A^2T^2 + \cdots = \sum_{m=0}^{\infty}\frac{1}{m!}A^mT^m \\
%     B_d & = \sum_{m=0}^{\infty}\frac{1}{(m+1)!}A^m B T^{m+1} \label{eq:B_discrete_2}
% \end{align}
Then, the corresponding discrete-time system can be obtained as 
\begin{align}\label{eq:linear_disc_sys}
    x_{k+1} = A_d x_k + B_d u_k + w_k,
\end{align}
where $w_k$ is the Gaussian process noise in the discretized system with zero mean and covariance $Q_d$
\begin{align}
    Q_d = \int_{kT}^{(k+1)T} e^{A((k+1)T-\tau)}Q e^{A^T((k+1)T-\tau)}d\tau .
\end{align}
Output of the discrete-time system can be obtained from Equation \eqref{eq:output_cont_sys} in the following form
\begin{align}\label{eq:output_disc_sys}
    y_k = Cx_k + v_k,
\end{align}
where $v_k$ is zero-mean Gaussian measurement noise with the covariance $R_d$, i.e.
\begin{align}
    v_k \sim \mathscr{N}(0,R_d).
\end{align}
$R_d$ should be calculated such that the covariance matrix of the discrete noise tends to the autocorrelation of the continuous noise as the sampling period tends to zero. In order to solve this problem, $1/T$ is used as an approximation for the Dirac delta function. The autocorrelation matrix defined by \eqref{eq:cont_measurement_cov} has infinite-valued elements. Then, by shrinking the sampling period and increasing the pulses amplitude, we have $R_d\rightarrow R/T$. As the sampling period tends to zero, the discrete noise sequence tends to one of the infinite-valued pulses of zero duration in \eqref{eq:cont_measurement_cov}. $R$ defined in \eqref{eq:cont_measurement_cov} is called the spectral density matrix. Then, the area under the impulse autocorrelation function in equivalent signal is $R_dT$ and equal to the area $R$ in the continuous white noise impulse autocorrelation \cite{gelb1974applied}. Then, the relationship between $R$ and $R_d$ can be obtained as
\begin{align}
    R_d \approx  \frac{R}{T} .
\end{align}

% One of the main assumptions which is made in the design of watermarking scheme in discrete-time systems is that the process and measurement noises are independent. In section \ref{sec:watermarking_cont_sys}, we will show that the process and measurement noises in a discrete-time system which is derived from a continuous-time system by sampling are not independent. Our focus in this paper is to study the effect of sampling period on the performance of watermarking framework in continuous-time systems. We will show that when the sampling period tends to zero, covariance of the process noise and measurement noise tends to zero. In other words, they will become independent, and by doing so, the main assumption which is made in developing watermarking method is satisfied. Therefore, to have independent noises and subsequently better performance for the estimator and controller, we need to decrease the sampling period. On the other hand, as we will discuss in section \ref{sec:watermarking_cont_sys}, the performance of the detector might degrade as the sampling period decreases. In order to solve the mentioned problem, we need to find an optimal value for sampling period. This value should be selected such that not only be the performance of the closed-loop control system acceptable, but also the watermarking framework be able to detect replay attacks in the system.

\section{Watermarking Framework in Discrete-time systems}\label{sec:watermarking_disc_sys}
In this section, a brief description is presented for the watermarking framework which has been developed for replay attack detection in discrete-time control systems \cite{mo2011cyber}. We assume that a hacker wants to corrupt the system defined in \eqref{eq:linear_disc_sys} and equipped with a Kalman filter and an LQG controller. We consider two main assumptions: 1) The attacker has access to all the sensors data; 2) The attacker can inject any control input to the system. To detect this attack, a watermarking signal is added to the normal control input. By doing so, we can distinguish between a normally operating system, which is driven by the current watermarks, and an attacked system, which is driven by a previous sequence of watermarks.

\subsection{Kalman Filter}
It is well known that the Kalman filter provides the optimal state estimate $\hat{x}_{k|k}$ for system \eqref{eq:linear_disc_sys}, as it provides the minimum variance unbiased estimate of the state $x_k$. 
\begin{align}
    \hat{x}_{0|-1} & = \Bar{x}_0, \quad P_{0|-1} = \Sigma \\
    \hat{x}_{k+1|k} & = A_d\hat{x}_k + B_du_k \nonumber \\
    P_{k+1|k} & = A_dP_kA_d^T + Q_d \nonumber \\
    K_k & = P_{k|k-1} C^T (CP_{k|k-1}C^T + R_d)^{-1} \nonumber \\
    \hat{x}_k & = \hat{x}_{k|k-1} + K_k(y_k - C\hat{x}_{k|k-1}) \nonumber \\
    P_k & = P_{k|k-1} - K_kCP_{k|k-1} \nonumber
\end{align}
When the usual conditions provided by Kalman, the estimator gain converges to its steady-state value, and in most of applications, this convergence occurs in a few steps. Therefore, we can write state error covariance, $P$, and Kalman gain, $K$, as follows:
\begin{align}
    P \overset{\Delta}{=} \lim_{k\rightarrow\infty} P_{k|k-1}, \quad K \overset{\Delta}{=} PC^T(CPC^T + R_d)^{-1} .
\end{align}
We assume the system to be in steady state. By considering $\Sigma = P$ as the initial condition, one can rewrite the Kalman filter as a fixed gain estimator in the following form:
\begin{align}
    \hat{x}_{0|-1} & = \Bar{x}_0, \quad \hat{x}_{k+1|k} = A_d\hat{x}_{k} + B_du_k \nonumber \\
    \hat{x}_k & = \hat{x}_{k|k-1} + K(y_k - C\hat{x}_{k|k-1}) .
\end{align}

\subsection{Linear Quadratic Gaussian (LQG) Control}
In this section, an LQG control scheme is designed such that  minimize the following infinite-horizon objective function
\begin{align}
    J = \lim_{N\rightarrow\infty}\mathbb{E}\frac{1}{N}\bigg[\sum_{k=0}^{N-1}(x_k^TWx_k + u_k^TUu_k) \bigg],
\end{align}
where $W$ and $U$ are positive semi-definite matrices. It is known that the above optimization problem yields the following fixed gain control input
\begin{align}\label{eq:lqg_optimal}
    u_k = u_k^* = -(B_d^TSB_d + U)^{-1}B_d^TSA_d\hat{x}_{k|k},
\end{align}
where $S$ can be obtained by solving the well-known infinite-horizon Riccati equation
\begin{align}
    S = A_d^TSA_d + W - A_d^TSB_d(B_d^TSB_d + U)^{-1}B_d^TSA_d .
\end{align}
By defining $L = -(B_d^TSB_d + U)^{-1}B_d^TSA_d$, the LQG controller can be rewritten as $u_k^* = L\hat{x}_{k|k}$. 
% Moreover, in our case, the objective function can be rewritten as below \textcolor{red}{do we need this?}
% \begin{align}
%     J = trace(SQ_d) + trace[(A_d^TSA_d+W-S)(P-KCP)].
% \end{align}

\subsection{\texorpdfstring{$\chi^2$}{TEXT} Failure Detector}
The $\chi^2$ detector has been widely used to detect anomalies in control systems. Principle idea of this detector is based on the probability distribution of the residual of Kalman filter.
\begin{theorem} \label{thm:kalman_residual} \cite{mehra1971innovations}
For the discrete-time system defined by \eqref{eq:linear_disc_sys} with Kalman filter and LQG controller, the residuals $y_i-C\hat{x}_{i|i-1}$ of Kalman filter are i.i.d. Gaussian distributed with zero mean and covariance $\mathscr{P}$, where $\mathscr{P} = CPC^T + R_d$.
\end{theorem}

By using Theorem \ref{thm:kalman_residual}, when the system is in normal condition, the probability to get the sequence $y_{k-\mathscr{T}+1} , \cdots, y_k$ can be obtained as follows
\begin{align}\label{eq:prob_detection}
    P(y_{k-\mathscr{T}+1} , \cdots, y_k) = \Bigg[\frac{1}{(2\pi)^{N/2}|\mathscr{P}|} \Bigg]^{\mathscr{T}}\exp\big(-\frac{1}{2}g_k\big),
\end{align}
where $\mathscr{T}$ is the window size of the detection, and $g_k$ is defined as
\begin{align}\label{eq:g_k_definition}
    g_k = \sum_{i=k-\mathscr{T}+1}^{k} (y_i - C\hat{x}_{i|i-1})^T\mathscr{P}^{-1}(y_i - C\hat{x}_{i|i-1}) .
\end{align}
When the probability of detection is low, one conclude that there is an anomaly in the system, but in order to apply $\chi^2$ detector, we do not need to calculate this probability. when the system is in normal condition, $g_k$ has a $\chi^2$ distribution with $m\mathscr{T}$ degrees of freedom. Then, One can use Equation \eqref{eq:g_k_definition} to detect any failure in the system. Therefore, the $\chi^2$ detector can be rewritten as 
\begin{align}\label{eq:detector}
    g_k \lessgtr threshold ,
\end{align}
where threshold is chosen for a specific false alarm probability. If $g_k$ is greater than the threshold, then the detector will trigger an alarm.

\subsection{Detection of Replay Attack in Control Systems}
In the replay attack, the attacker resends a set of data which is recorded for a period of time. Therefore, the attacked system can be modeled by a virtual system in the following form
\begin{align}
    x'_{k+1} & = A_d x'_k + B_d u'_k, y'_k = Cx'_k + v'_k \nonumber \\
    \hat{x}'_{k+1|k} & = A_d \hat{x}'_{k|k} + B_d u'_k \nonumber \\ \hat{x}'_{k+1|k+1} & = \hat{x}'_{k+1|k} + K(y'_k - \hat{x}'_{k+1|k}) \nonumber \\
    u'_k & = L\hat{x}'_{k|K} .
\end{align}
The virtual system is a shifted version of the main system. State estimation for the main and virtual systems can be represented by the following equations
\begin{align}
    \hat{x}_{k+1|k} & = (A_d+B_dL)(I-KC)\hat{x}_{k|k-1} + (A_d+B_dL)Ky'_k \nonumber\\
    \hat{x}'_{k+1|k} & = (A_d+B_dL)(I-KC)\hat{x}'_{k|k-1} + (A_d+B_dL)Ky'_k .
\end{align}
Residuals of Kalman filter in the main and virtual systems have the same distribution. Let us define $\mathscr{A}\overset{\Delta}{=}(A_d+B_dL)(I-KC)$ and $\zeta\overset{\Delta}{=}\hat{x}_{0|-1}-\hat{x}'_{0|-1}$, then $g_k$ for an attacked system can be obtained as
\begin{align}\label{eq:g_k_under_attack}
    g_k = & \sum_{i=k-\mathscr{T}+1}^{k} \Big[(y'_i - C\hat{x}'_{i|i-1})^T\mathscr{P}^{-1}(y'_i - C\hat{x}'_{i|i-1}) \nonumber\\
    & + 2(y'_i - C\hat{x}'_{i|i-1})^T\mathscr{P}^{-1}C\mathscr{A}^i\zeta \nonumber\\
    & + \zeta^T(\mathscr{A}^i)^TC^T\mathscr{P}^{-1}C\mathscr{A}^i\zeta\Big] .
\end{align}
To evaluate the performance of designed $\chi^2$ detector, we need to consider two cases. 
\begin{enumerate}
    \item $\mathscr{A}$ is stable: In this case, the second and third terms in \eqref{eq:g_k_under_attack} will converge to zero. Hence, $g_k$ for the main and virtual system  has the same distribution. Then, the detector is completely useless to detect any replay attack in the control system. 
    \item $\mathscr{A}$ is unstable: Any replay attack which is applied for a long time can be detected by the $\chi^2$ detector, because $g_k$ will soon become unbounded, and by comparing $g_k$ for the main and virtual system, the replay attack will be detected. 
\end{enumerate}
We conclude that the $\chi^2$ detector is useful only for an unstable $\mathscr{A}$. To be able to detect replay attack when $\mathscr{A}$ is stable, the control input is redesigned by adding an authentication signal. Let us define the new controller in the following form:
\begin{align}
    u_k = u_k^{*} + \Delta u_k ,
\end{align}
where $u_k^{*}$ is the LQG optimal control which is defined by \eqref{eq:lqg_optimal}, and $\Delta u_k$ is an authentication signal added to the optimal control to be able to detect replay attack. The sequence $\Delta u_k$ is drawn from an i.i.d. Gaussian distribution with zero mean and covariance $\mathscr{Q}$, and independent of $u_k^{*}$. By adding this authentication signal, the control input will not be the optimal one. However, it will help us to detect replay attack. In other words, the watermarking framework sacrifice the control performance to be able to detect the attack. 

\begin{theorem}\cite{mo2013detecting}\label{thm:LQG_performance}
After adding the authentication signal to the optimal control, the LQG performance is given by
\begin{align}\label{eq:new_lgq_cost}
    J' = J + trace[(U + B_d^TSB_d)\mathscr{Q}].
\end{align}
\end{theorem}

The following theorem represents performance of the $\chi^2$ detector in presence of replay attack.
\begin{corollary}\cite{mo2013detecting}
In the absence of an attack, the expectation of $g_k$ in the $\chi^2$ detector is 
\begin{align}\label{eq:g_k_without_attack}
    \mathbb{E}[g_k] = m\mathscr{T} .
\end{align}
Under attack, the asymptotic expectation becomes
\begin{align}\label{eq:g_k_with_attack}
    \lim_{k\rightarrow\infty}\mathbb{E}[g_k] = m\mathscr{T} + 2trace(C^T\mathscr{P}^{-1}C\mathscr{U})\mathscr{T} ,
\end{align}
\end{corollary}
where $\mathscr{U}$ is the solution to the following equation
\begin{align}\label{eq:relation_cU_cQ}
    \mathscr{U} = \sum_{i=0}^{\infty}\mathscr{A}^iB_d\mathscr{Q}B_d^T(\mathscr{A}^i)^T .
\end{align}

The difference in the expectations of $g_k$ with and without attack proves that the detection rate does not converge to the false alarm rate.

\section{Watermarking in Continuous-time Systems}\label{sec:watermarking_cont_sys}

In this section, the effect of the sampling period on performance of the watermarking framework is studied. In digital signal processing, the sampling theorem states that to reconstruct an unknown band-limited signal from discretized version of that signal, the sampling rate must be at least twice as high  as the highest frequency in the signal. In digital control, this theorem is applied to a feedback controller. Thus, based on this theorem the sampling rate must be at least twice the required closed-loop bandwidth of the system. 
In most of applications, to get an appropriate time response for the control system, this sampling rate would be inadequate. Moreover, one of the most important concepts which should be considered is the delay between a command input and the system response to the command input. This delay should be reduced as much as possible. In order to confront these issues, Franklin G.F. \textit{et al.} \cite{franklin1998digital} proposed that the sampling rate should be at least $20$ times the required closed-loop bandwidth of the system. This is a lower bound for sampling rate. Therefore, this shows that decreasing sampling period will increase the control performance. However, to evaluate the performance of watermarking framework, we need to analyze performance of controller and detector simultaneously. In the next subsection, we will analyze the effect of sampling period on the $\chi^2$ detector performance.

\subsection{Effect of Sampling Period on the \texorpdfstring{$\chi^2$}{TEXT} Detector}

Sampling period affects not only the control system response but also performance of the detector. As it is mentioned, if the sampling period decreases, time response of the control system will improve. The effect of small sampling period on performance of the $\chi^2$ detector is studied in the following lemma and corollary.

\begin{lemma}\label{lemma:delta_g_k}
Consider system \eqref{eq:linear_cont_sys} with the LQG controller \eqref{eq:lqg_optimal} designed on the sampled data system \eqref{eq:linear_disc_sys}, assuming that $\mathscr{U}$ be bounded for any $T>0$, then as $T$ goes to zero, the detector \eqref{eq:detector} is not able to detect any replay attack.
% Consider a continuous-time stochastic system which is defined by \eqref{eq:linear_cont_sys} with an LQG control and $\chi^2$ detector. It is assumed that the controller and detector are applied to the system by using ZOH method. Then, by choosing very small sampling period, $T\rightarrow 0$, and considering the fact that the difference of LQG performance with and without watermarking signal is bounded, the $\chi^2$ detector will be  useless to detect any replay attack in the control system.
\end{lemma}

% the difference in the expectation of $g_k$ with and without attack is obtained as

\begin{proof}
By using a very small sampling period, i.e. $T\rightarrow 0$, the covariance of the measurement noise in the discretized system can be approximated by $R_d\approx R/T$. By substituting this approximation in the covariance of the residual of Kalman filter, its covariance can be rewritten as
\begin{align}\label{eq:cov_kalman_substit}
    \mathscr{P} = CPC^T + \frac{R}{T} .
\end{align}
All the terms in the Taylor expansions of $A_d$, $B_d$, and $Q_d$ are constant or $\mathcal{O}(T)$\footnote{$g(T) = \mathcal{O}(T)$ as $T\rightarrow 0$: ''asymptotically $g$ goes to zero at least as fast as $T$``, or more formally:\newline
$\quad\exists K\geq 0\quad \mbox{s.t.} \quad \Big|\frac{g(T)}{T}\big|\leq K \quad \mbox{as} \quad T\rightarrow 0$.}. Because of the presence of $(CP_{k|k-1}C^T + R/T)^{-1}$ in $K_k$ and by considering the fact that we use very small sampling period, this term can be considered an $\mathcal{O}(T)$. Given that the Kalman filter does not have any term $1/T$, the matrix $P$ does not contain $1/T$. Since we assumed that the sampling period is very small, the second term in \eqref{eq:cov_kalman_substit} will be the dominant term. Thus, the covariance of the residuals of the Kalman filter and its inverse can be approximated as below
\begin{align}\label{eq:cov_kalman_approx}
    \mathscr{P} \approx \frac{R}{T} , \quad
    \mathscr{P}^{-1} \approx TR^{-1} .
\end{align}
By substituting $\mathscr{P}^{-1}$ from \eqref{eq:cov_kalman_approx} into \eqref{eq:g_k_with_attack}, the difference in the expectation of $g_k$ with and without attack can be written as
\begin{align}\label{eq:delta_g_k_t}
    \mathbb{E}[\Delta g_k] = 2trace(C^TR^{-1}C\mathscr{U})\mathscr{T}T.
\end{align}
Therefore, as the sampling period tends to zero, the difference in the expectation of $g_k$ with and without attack tends to zero, and according to \eqref{eq:g_k_with_attack} of Corollary 1, the detector is not able to detect any replay attack. 
\end{proof}
\begin{remark}
As the sampling period goes to zero, the replay attack detection rate decreases. On the other hand, by increasing the sampling period, $\mathbb{E}[\Delta g_k]$ decreases due to the degradation of the control performance which results in larger steady state error and consequently in a lower detection rate. Then, there is an optimal value of the sampling period that maximizes detection performance. 
% According to Lemma \ref{lemma:delta_g_k}, as $T$ goes to zero, the detector is not able to detect attack unless we increase all the elements of $\mathscr{U}$ inversely proportional by $T$. The closed-loop system is stable and $\mathscr{A}$ is bounded and its eigenvalues are inside the unit circle, and from Equation \eqref{eq:B_discrete_2}, we know that by decreasing the sampling period, all the elements of matrix $B_d$ decrease. We showed that $\mathscr{U}$ should be increased inversely proportional by $T$, otherwise we will not be able to detect replay attack. Then, according to Equation \eqref{eq:relation_cU_cQ}, covariance of watermarking signal, and subsequently the difference of LQG cost with and without watermarking signal will increase. We know elements of matrix $B_d$ decreases by decreasing the sampling period. Even though $B_d$ is very small, because of presence of matrix $U$ in Equation \eqref{eq:new_lgq_cost}, increasing the covariance of watermarking signal will definitely increase $\Delta J$, and subsequently performance of the LQG controller will degrade. Therefore, as the sampling period goes to zero, the $\chi^2$ detector is not able to detect replay attack.
\end{remark}

In the next subsection, we proposed an optimization to manage the trade-off between detection rate and control performance. 
% In future works, we are going to propose a method to optimize the sampling period with respect to the performance of the watermarking framework.

\subsection{Optimization}
We have shown in the previous subsection that the sampling period is a compromise between the performance of the controller and detector. Therefore, to find the optimal value of the sampling period and covariance of watermarking signal, we have to solve the following optimization problem.
\begin{align}\label{eq:optimization_problem_1}
    \max_{\mathscr{Q},T} & \quad 2trace(C^T\mathscr{P}^{-1}C\mathscr{U})\mathscr{T} \\
    \text{subject to} & \quad trace[(U + B_d^TSB_d)\mathscr{Q}] < \mu \nonumber \\
    & \quad 0 < T \leq \Bar{T} \nonumber \\
    & \quad \mathscr{U} - B_d\mathscr{Q}B_d^T = \mathscr{A}\mathscr{U}\mathscr{A}^T, \nonumber
\end{align}
where $\mu$ is the extra cost due to the watermarking signal, and $\Bar{T}$ is upper bound of the sampling period which is used to prevent aliasing and to give a proper system response. Since all the parameters of the system, controller, and detector are functions of the sampling period the optimization problem \eqref{eq:optimization_problem_1} may be hard to solve. To simplify this procedure we fix the sampling period $T$ in order to compute the following optimization problem
\begin{align}\label{eq:optimization_problem_2}
    \max_{\mathscr{Q}} & \quad 2trace(C^T\mathscr{P}^{-1}C\mathscr{U})\mathscr{T} \\
    \text{subject to} & \quad trace[(U + B_d^TSB_d)\mathscr{Q}] < \mu \nonumber \\
    & \quad \mathscr{U} - B_d\mathscr{Q}B_d^T = \mathscr{A}\mathscr{U}\mathscr{A}^T \nonumber.
\end{align}
Then, we iterate the computation of \eqref{eq:optimization_problem_2} for different values of $T$. Bisection methods can be used to find the optimal sampling period $T$. 
To have a fair comparison, a constant value $\mu$ is used as the extra cost for all the sampling period values.

%and considering that we want to analyze the effect of sampling period on the framework performance, to solve this problem, we choose different values for the sampling period, and for each of them, we solve this optimization problem using CVX \cite{grant2014cvx}.  The optimization problem for each sampling period can be rewritten as

%For each sampling period, after solving the optimization problem, expectation of $\Delta g_k$ can be calculated using \eqref{eq:g_k_with_attack}. Finally, we can choose the sampling period which prevents aliasing in the control system with highest expectation of $\Delta g_k$.

\section{Simulation Results}\label{sec:simulation}

In this section, the theoretical results is evaluated through intensive simulation studies carried out to detect a replay attack in a quadrotor. First, we will represent the mathematical model of a quadrotor. Then, the watermarking framework will be applied to this system. 

Dynamical behavior of a quadrotor can be modeled by nonlinear differential equations \cite{beard2008quadrotor}. Since the watermarking framework is developed for a system in steady-state, we assume that the quadrotor is hovering around an equilibrium point. Then, we can linearize its nonlinear model around this point. After linearization, the system can be represented in state space form of \eqref{eq:linear_cont_sys}. Matrices $A$ and $B$ are presented in \cite{beard2008quadrotor}. The state variables vector is
\begin{align}
    x = [ \Dot{p}_x \quad p_x \quad \Dot{p}_y \quad p_y \quad \Dot{p}_z \quad p_z \quad \Dot{\phi} \quad \phi \quad \Dot{\theta} \quad \theta \quad \Dot{\psi} \quad \psi ]^T,
\end{align}
where $p_x$, $p_y$, and $p_z$ are used to determine position of the quadrotor in three principal directions; and $\phi$, $\theta$, and $\psi$ are the roll, pitch and yaw angles, respectively; The output of the system is defined as $y = [ p_x \quad p_y \quad p_z \quad \psi ]^T$; and Control input is defined as follows
\begin{align}
    u = \begin{bmatrix}
    F & \tau_{\phi} & \tau_{\theta} & \tau_{\psi}
    \end{bmatrix}^T,
\end{align}
where $F$ is the force that acts on the quadrotor, and $\tau_{\phi}$, $\tau_{\theta}$, and $\tau_{\psi}$ are rolling, pitching, and yawing torques. Physical parameters of the quadrotor including mass and moments of inertia are set to $m = 0.6 kg$, $J_x = J_y = 0.0092 kgm^2$, and $J_z = 0.0101kgm^2$. The process and measurement noises are considered as independent Gaussian random variables with zero mean. 

% Figure \ref{fig:delta_gk} shows the expectation of $\Delta g_k$ for different sampling periods. The optimal value we found is $T = 0.1 s$. 
The optimal value of the sampling period for our system is $T = 0.1 s$. Figure \ref{fig:delta_gk} shows the expectation of $\Delta g_k$ for different sampling periods. The optimization problem \eqref{eq:optimization_problem_2} is solved by using CVX \cite{grant2014cvx}. It can be observed that when the sampling period is very small, $\mathbb{E}[\Delta g_k]$ becomes very small, then by increasing the sampling period, $\mathbb{E}[\Delta g_k]$ increases, but after $T=0.1$, it starts decreasing.

To show the effect of sampling period on the detector's performance, the watermarking framework is applied to the quadrotor with different sampling periods. Simulation results are shown in Figure \ref{fig:ROC_Curve}. 
This figure shows several ROC curves for the quadrotor system under a replay attack. Six different sampling periods are used, and it can be observed that by increasing the sampling period, performance of the detector improves. However, after $T=0.1s$ the ROC curve starts to change in the reverse direction and detector performance degrades. This behavior of the detector was expected, because the maximum value of $\mathbb{E}[\Delta g_k]$ occurs in $T=0.1s$.
To show the effect of sampling period on the LQG performance, cost function for these sampling periods are calculated and the lowest sampling period, $T = 0.01 s$ is considered as the reference cost function ($J_1$) and the ratio of the cost function for other sampling periods ($J_T$) to this reference value is shown in Table \ref{tab:obj_func}.

\begin{table}[!h]
    \centering
    \begin{tabular}{ |c|c|c|c|c|c| }
        \hline
        $T$ & 0.02 & 0.04 & 0.07 & 0.10 & 0.15  \\
        \hline
        $J_T/J_1$ & 1.03 & 1.10 & 1.21 & 1.38 & 1.68 \\ 
        \hline
    \end{tabular}
    \caption{lqg cost function for different values of sampling period.}
    \label{tab:obj_func}
\end{table}

% \begin{table}[]
%     \centering
%     \begin{tabular}{ |c|c|c|c|c|c|c|c|c| }
%         \hline
%         $T$ & 0.02 & 0.03 & 0.04 & 0.05 & 0.06 & 0.07 & 0.08 & 0.09  \\
%         \hline
%         $J_T/J_1$ & 1.03 & 1.07 & 1.12 & 1.16 & 1.21 & 1.26 & 1.32 & 1.37  \\ 
%         \hline
%     \end{tabular}
%     \caption{lqg cost function for different values of sampling period.}
%     \label{tab:obj_func}
% \end{table}

\begin{figure}[!t]
\centering
\includegraphics[width=9cm]{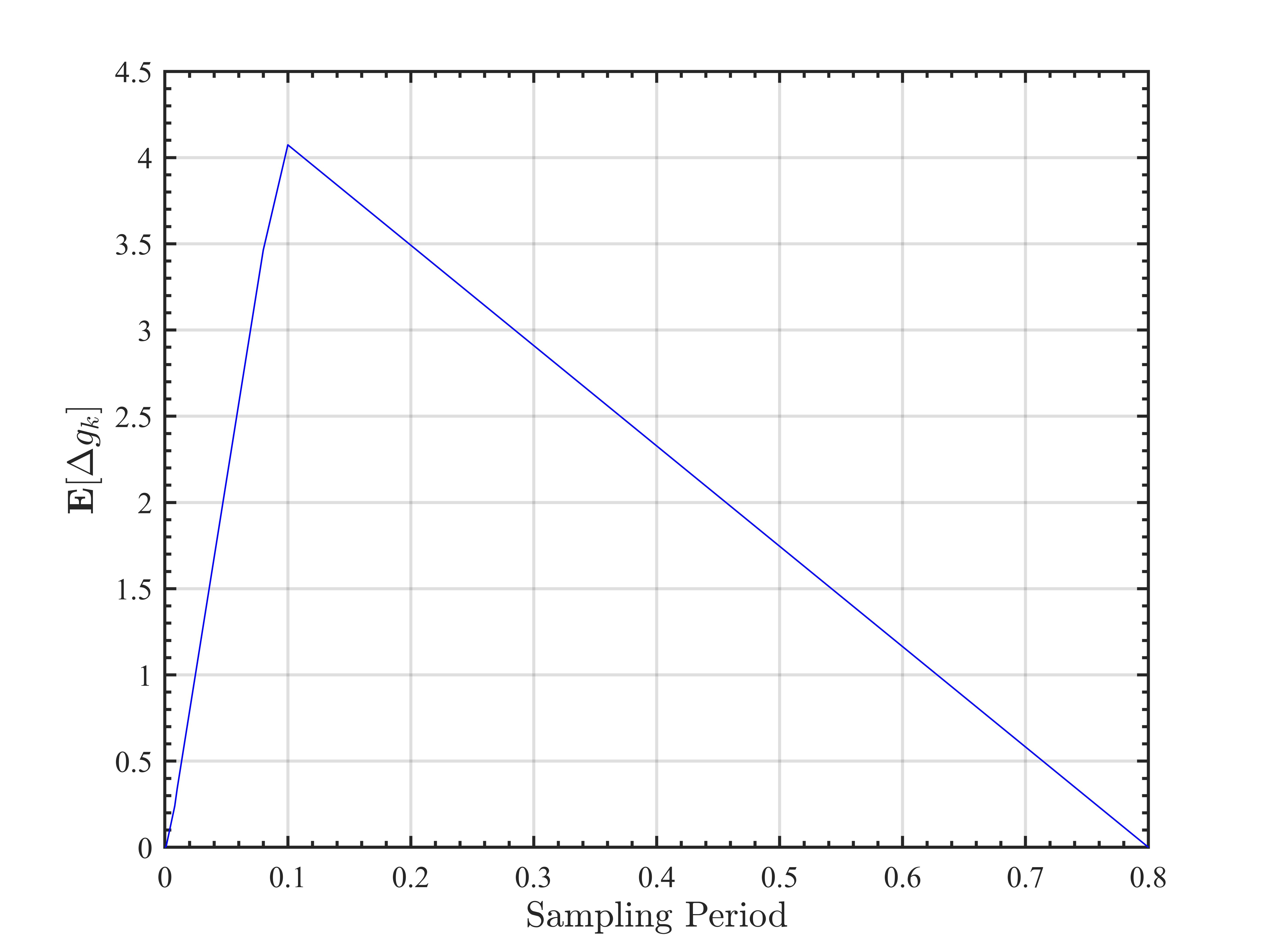} \\
\caption{Expectation of $\Delta g_k$ for different sampling periods.}
\label{fig:delta_gk}
\end{figure}

\begin{figure}[!t]
\centering
\includegraphics[width=9cm]{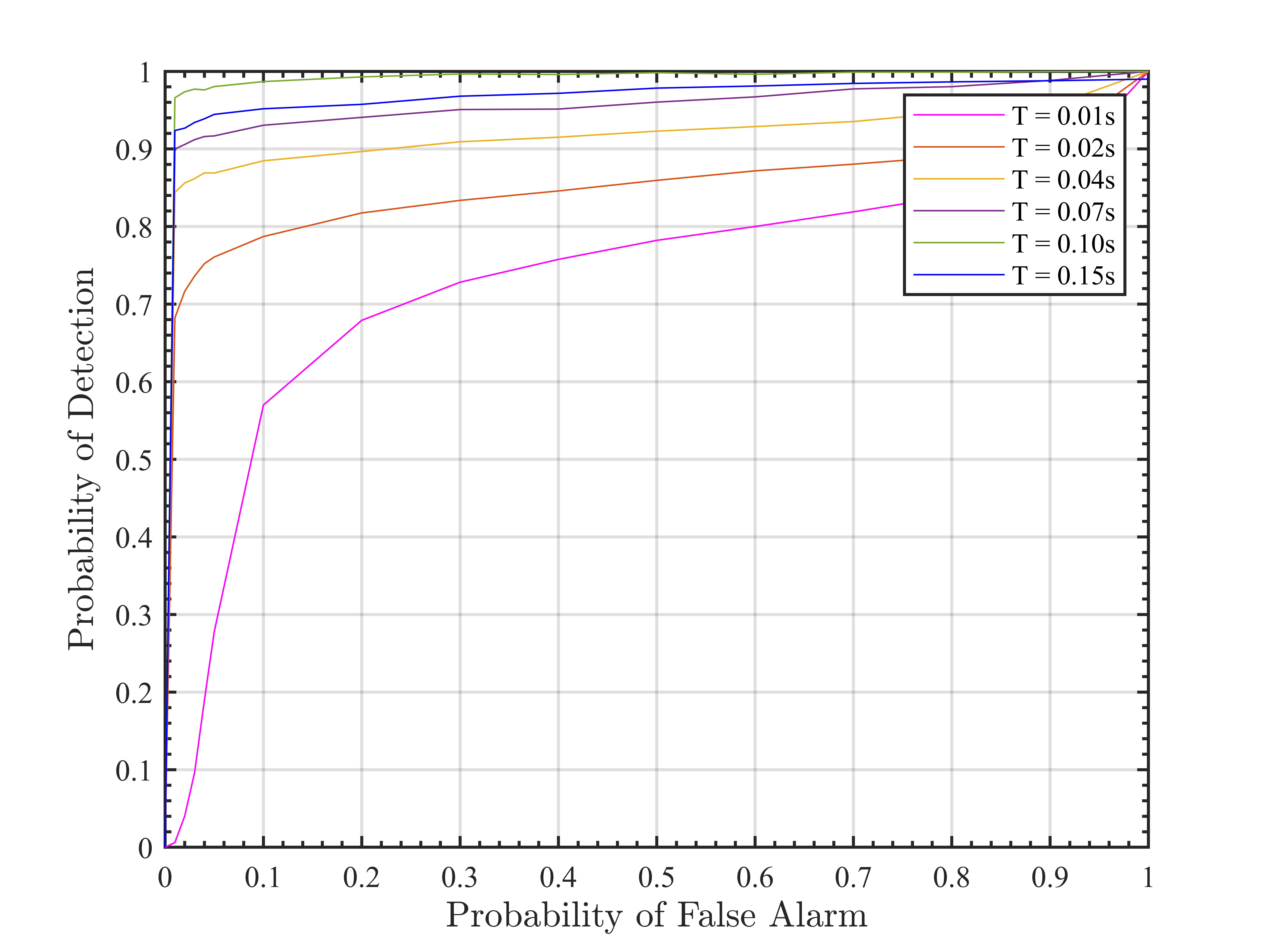} \\
\caption{ROC curve for replay attack detection in the quadrotor system with different sampling periods.}
\label{fig:ROC_Curve}
\end{figure}

\section{Conclusion}
In this paper, we generalize the design of physical watermarking to detect replay attack on  digitally controlled continuous-time systems. In particular, we investigate the effect of sampling period on the performance of the detector. We show that an optimal sampling period exists and we generalize the optimal watermarking signal design to include sampling period as a design variable jointly with the covariance of the watermark. Finally, we apply the watermarking framework to a quadrotor, and numerical simulations are included to illustrate our findings and validate our design.

%---------------------------------------------------------------------------------------------------------
% \balance
\bibliographystyle{IEEEtran}
\bibliography{ref}{}

%---------------------------------------------------------------------------------------------------------
%---------------------------------------------------------------------------------------------------------
\end{document}